\documentclass{article}
\usepackage[T1]{fontenc}
\usepackage{lmodern}
\usepackage{graphicx}
\usepackage{float}

\usepackage{amsmath}
\usepackage{amsthm}
\usepackage{amssymb}
\usepackage{authblk}
\usepackage{url}
\usepackage[noadjust]{cite}
\usepackage[hidelinks]{hyperref}
\usepackage{cleveref}
\usepackage{fullpage}

\newtheorem{theorem}{Theorem}
\newtheorem{lemma}[theorem]{Lemma}
\newtheorem{definition}[theorem]{Definition}

\newtheorem{observation}[theorem]{Observation}

\usepackage{xcolor}

\newcommand{\prob}[1]{\textsf{\textsc{#1}}}

\newcommand{\re}{\mathit{RE}}
\newcommand{\sre}{\mathit{SRE}}

\renewcommand{\S}[1]{\mathsf{S}_{#1}}
\newcommand{\T}[1]{\mathsf{T}_{#1}}
\newcommand{\A}[1]{\mathsf{A}_{#1}}
\newcommand{\B}[1]{\mathsf{B}_{#1}}
\newcommand{\X}[1]{x_{#1}}
\newcommand{\Eq}{\mathsf{Eq}}
\newcommand{\Ch}{\mathsf{Ch}}
\newcommand{\EqS}{\Eq_{\mathsf{S}}}
\newcommand{\EqT}{\Eq_{\mathsf{T}}}
\newcommand{\EqA}{\Eq_{\mathsf{A}}}
\newcommand{\EqB}{\Eq_{\mathsf{B}}}
\newcommand{\EqX}{\Eq_{\mathsf{X}}}
\newcommand{\EqXv}[1]{\Eq_{\mathsf{X}}(#1)}

\title{String Representation Based on Substring Equation Systems}

\author[1]{Hiroki Shibata\thanks{\texttt{shibata.hiroki.753@s.kyushu-u.ac.jp}}}
\affil[1]{Joint Graduate School of Mathematics for Innovation, Kyushu University, Japan}
\author[2]{Hideo Bannai\thanks{\texttt{bannai.h.557a@m.isct.ac.jp}}}
\affil[2]{M\&D Data Science Center, Institute of Biomedical Engineering, Institute of Science Tokyo, Japan}

\makeatletter
\@ifundefined{mathindent}{\newlength{\mathindent}}{}
\makeatother
\begin{document}

\maketitle

\begin{abstract}
Repetitiveness measures quantify how much repetitive structure a string contains and serve as parameters for compressed representations and indexing data structures.
Many compression schemes represent strings by recording equalities between identical substrings.
We study string compression based on \emph{substring equation systems (SESs)}, which represent strings as unique solutions to substring-equality and character-assignment constraints.
We show that every string $w$ has an SES of size $O(\chi(w))$, where $\chi(w)$ is the size of a \emph{smallest suffixient set} for $w$ with a unique terminal symbol appended.
This result establishes the reachability of $\chi$, which had been an open problem.
We also prove that computing the size $s(w)$ of the smallest SES that represents $w$ is NP-hard and $(1 + \varepsilon)$-inapproximable for some fixed constant $\varepsilon > 0$.
Finally, we prove that the size $b(w)$ of the smallest bidirectional macro scheme (BMS) representing $w$ satisfies $s(w) \leq b(w) \leq 4s(w)$.
Hence, SES and BMS are equivalent up to a constant factor, and this equivalence gives the new bound $b(w) \in O(\chi(w))$.
\end{abstract}

\section{Introduction}

Repetitiveness measures aim to quantify how much repetitive structure a string contains.
They offer a unified way to compare strings and to reason about the space usage of compressed representations.
Beyond repetitiveness measures that arise directly from compression techniques, such as LZ-style parsing~\cite{DBLP:journals/tit/ZivL77}, repetitiveness measures that are not derived from any specific compression technique have been proposed, notably string attractors~\cite{DBLP:conf/stoc/KempaP18} and normalized substring complexity~\cite{DBLP:conf/latin/KociumakaNP20}.
These measures are widely used to analyze the size of existing compressed representations and indices, and they provide convenient parameters for stating space bounds for both representations and indexing data structures~\cite{DBLP:journals/csur/Navarro21a}.

Many compression schemes represent strings compactly by capturing equalities between identical substrings.
LZ-style parsings~\cite{DBLP:journals/tit/ZivL77} are a fundamental example of this approach.
Among general representation models based on substring references, the bidirectional macro scheme (BMS)~\cite{DBLP:conf/stoc/StorerS78} captures a broad class of compression schemes.
A BMS partitions the string into phrases and stores each phrase either explicitly as a character or by a directed reference to an identical substring.
Although BMS is known to be theoretically stronger than many other compression schemes, it has several structural restrictions.
Specifically, its references must be acyclic, each copied phrase has only one source, and its phrase boundaries form a partition of the text.
These restrictions must be handled when constructing BMSs for string families, designing algorithms for BMS, and establishing properties of BMS.

In this paper, we study the use of \emph{substring equation systems} (SESs) as a general approach to string compression.
An SES represents a string as the unique solution to a collection of substring-equality constraints of the form $w[i..i+\ell-1] = w[j..j+\ell-1]$ and character-assignment constraints $w[i] = c$.
While previous studies on substring equations have mainly focused on finding solutions and constructing compressed representations of these solutions~\cite{DBLP:journals/tcs/GawrychowskiKRR20,DBLP:conf/soda/KociumakaNW25,AgolHT06},
we regard an SES with a unique solution as a compressed representation.
We study the minimum number $s(w)$ of constraints that uniquely determine a string $w$ as a repetitiveness measure and compare it with other repetitiveness measures.

The definition of SESs naturally implies that every BMS can be converted into an SES of the same size.
This conversion directly gives $s(w) \leq b(w)$ for every string $w$, where $b(w)$ is the size of the smallest BMS representing $w$.
This inequality shows that SES is at least as powerful as BMS for compressing strings.
Unlike BMSs, an SES does not require the string to be partitioned into phrases, and its substring-equality constraints are undirected.
Together, these properties make SES a more flexible representation framework than BMS.

One of our main results is that every string $w$ has an SES of size $O(\chi(w))$, where $\chi(w)$ is the size of a \emph{smallest suffixient set} for $w$ with a unique terminal symbol appended~\cite{DBLP:journals/corr/abs-2312-01359}.
Suffixient sets were first introduced as a basis for space-efficient string indexes~\cite{DBLP:journals/corr/abs-2312-01359,DBLP:journals/mst/CenzatoDGKMOP26}, and their minimum size $\chi$ is computable in linear time~\cite{DBLP:conf/spire/CenzatoOP24,DBLP:conf/spire/NavarroRU25}.
Existing results already place $\chi$ among the central repetitiveness measures~\cite{DBLP:conf/spire/NavarroRU25}.
For example, $\chi(w) \leq 2r(w)$ and $\chi(w) \leq 2\bar{r}(w)$ hold for every string $w$, where $r(w)$ and $\bar{r}(w)$ are the numbers of runs in the Burrows--Wheeler transforms~\cite{Burrows1994ABL} of $w$ and its reversal, respectively~\cite{DBLP:conf/spire/NavarroRU25,DBLP:journals/corr/abs-2312-01359}.
It is also known that $\chi$ is bounded by a constant for finite episturmian words over a fixed alphabet, while it is unknown whether constant-size BMS representations exist even for Tribonacci words, which are simple examples of finite episturmian words~\cite{DBLP:conf/spire/NavarroRU25}.
This gap suggests that suffixient sets capture repetitive structures that are not easily represented by BMS.
However, it remained open whether $\chi$ is reachable, that is, whether every string $w$ admits an $O(\chi(w))$-word representation.
The reachability of $\chi$ is also related to a long-standing open problem concerning string attractors.
The size $\gamma$ of the smallest string attractor satisfies $\gamma(w) \leq \chi(w)$~\cite{DBLP:conf/spire/NavarroRU25}, but whether $\gamma$ is reachable remains open.
Regarding the reachability of $\chi$, Navarro et al.~\cite{DBLP:conf/spire/NavarroRU25} state the conjecture: ``We conjecture, instead, that $\chi$ is not reachable, proving which would imply that $\gamma$ is also unreachable, a long-time open question.''
We disprove this conjecture by constructing an SES of size $O(\chi(w))$ for every string $w$.
Such an SES can be constructed in $O(|w|)$ time over a linearly sortable alphabet.

We also study the complexity of computing $s(w)$, the size of the smallest SES that represents $w$.
We show that computing $s(w)$ and approximating $s(w)$ within a factor of $1 + \varepsilon$ are NP-hard for some fixed constant $\varepsilon > 0$. 

Finally, we prove that every SES representing $w$ can be converted into a BMS for $w$ without an asymptotic increase in size.
More precisely, $s(w)$ and $b(w)$ satisfy $s(w) \leq b(w) \leq 4s(w)$.
We also obtain the new bound $b(w) \in O(\chi(w))$ by combining this inequality with our $O(\chi(w))$-size SES construction.
Since $\chi$ is bounded by a constant for finite episturmian words over a fixed alphabet, this bound gives the first known construction of constant-size BMS representations for these words.

\subsection*{Related Work}
Prior to this work, there have been several studies on systems of substring equations.
Gawrychowski et al.~\cite{DBLP:journals/tcs/GawrychowskiKRR20} gave an algorithm that computes a generic solution to a system of substring equations and runs in linear time in the number of constraints and the length of the output string.
Kociumaka et al.~\cite{DBLP:conf/soda/KociumakaNW25} considered representing all strings satisfying the substring equations by an extended SLP in which each free variable represents a single character.
They gave an algorithm that constructs such an extended SLP of size $O(|\Eq| \log^4 n)$, where $n$ denotes the length of the string and $\Eq$ denotes the set of substring equations.
Substring equations have also been studied in the context of computational topology.
Agol et al.~\cite{AgolHT06} gave an algorithm that computes the maximum number of distinct characters in a string satisfying the substring equations.
These results mainly concern solving given systems of substring equations and representing their solutions compactly, whereas we study the construction of a small system whose unique solution is a given string.

Substring equations have also been studied as a special case of word equations.
The word equation problem is to assign a string to each variable so that all given equality constraints are satisfied.
Substring equations are a special case of word equations in which the length of each variable is fixed in advance.
For this setting, Plandowski and Rytter~\cite{DBLP:conf/icalp/PlandowskiR98} gave an algorithm to construct a compressed representation of a solution.

Macro schemes and bidirectional macro schemes are universal compression schemes based on equalities between substrings.
A macro scheme represents a string by a sequence of terminal characters and pointers that is decoded according to prescribed rules.
Storer and Szymanski~\cite{DBLP:journals/jacm/StorerS82} introduced four basic models, of which the OPM model corresponds to bidirectional macro schemes.
They also compared these models in terms of the sizes of their smallest representations.
However, all these models are based on directed reference relations using pointers, and none of them are exactly equivalent to SES.

Some papers distinguish general macro schemes from BMSs and describe the former as a set of substring equations that uniquely determines the string~\cite{DBLP:conf/esa/KempaPPR18,DBLP:conf/stoc/KempaP18}.
For example, Kempa and Prezza~\cite{DBLP:conf/stoc/KempaP18} define a macro scheme as a set of rules of the forms $T[i..j] \leftarrow T[i'..j']$ and $T[i] \leftarrow c$ such that the string can be unambiguously decoded.
With this definition, the minimum size $m(w)$ of a general macro scheme satisfies $s(w) \leq m(w) \leq 2s(w)$.
Indeed, each SES constraint $T[i..j] = T[i'..j']$ can be replaced by the two rules $T[i..j] \leftarrow T[i'..j']$ and $T[i'..j'] \leftarrow T[i..j]$ while preserving the uniqueness of the solution.
However, it was unclear before this work whether $m(w)$ is within a constant factor of $b(w)$.

\section{Preliminaries}
For every nonnegative integer $n$, let $[n] = \{1, \dots, n\}$, where $[0] = \emptyset$.
For integers $a \leq b$, let $[a, b] = \{i \in \mathbb{Z} \mid a \leq i \leq b\}$ and $[a, b) = \{i \in \mathbb{Z} \mid a \leq i < b\}$.
Let $\Sigma$ be an alphabet of size $|\Sigma| \geq 2$.
An element of $\Sigma$ is called a character.

A string $w$ of length $|w| = n$ over the alphabet $\Sigma$ is a sequence $w[1]\cdots w[n]$ of characters where $w[i] \in \Sigma$ for all $1 \leq i \leq n$.
For any two strings $x$ and $y$, we denote by $xy = x[1] \cdots x[|x|] y[1] \cdots y[|y|]$ the concatenation of $x$ and $y$.
For any string $x$, we denote by $x^R = x[|x|] \cdots x[1]$ its reversal.
If $w = xyz$ holds for some strings $x, y, z \in \Sigma^*$, then $x$, $y$, and $z$ are called a prefix, a substring, and a suffix of $w$, respectively.
For $1 \leq i \leq j \leq n$, we denote by $w[i..j] = w[i]\cdots w[j]$ the substring of $w$ from position $i$ to $j$.
For convenience, $w[i..j]$ is the empty string when $j < i$.
For every string $w$, we denote by $\sigma(w)$ the number of distinct characters appearing in $w$.

Let $\prec$ be a total order on $\Sigma$.
The \emph{lexicographic order} on $\Sigma^*$ defined by $\prec$ is given by $X \prec Y$ if and only if $X$ is a proper prefix of $Y$ or there exists an integer $1 \leq i \leq \min\{|X|,|Y|\}$ such that $X[1..i-1]=Y[1..i-1]$ and $X[i] \prec Y[i]$ for any strings $X, Y \in \Sigma^*$.

For any $x\in\Sigma^*$ and $c\in \Sigma$, 
$xc$ is a \emph{right extension} in $w$ if
$xc$ is a substring of $w$ and for some
$c' \in \Sigma \setminus \{ c \}$, $xc'$ is also a substring of $w$.
A right extension in $w$ is \emph{super-maximal} if it is not a proper suffix of another right extension in $w$.
The sets of all right extensions and super-maximal right extensions in $w$ are denoted by $\re(w)$ and $\sre(w)$, respectively.

Existing work establishes a one-to-one correspondence between smallest suffixient sets and super-maximal right extensions~\cite{DBLP:journals/mst/CenzatoDGKMOP26,DBLP:conf/spire/CenzatoOP24}.
Using the notation of Navarro et al.~\cite{DBLP:conf/spire/NavarroRU25}, we adopt the following definition.
\begin{definition} \label{def:smallest_suffixient_set}
For a string $w$ of length $n$,
a set $S \subseteq [n]$ is a \emph{smallest suffixient set} for $w$ if
there exists a bijection $f: \sre(w) \rightarrow S$ such that every super-maximal right extension $x \in \sre(w)$ is a suffix of $w[1..f(x)]$.
\end{definition}
The repetitiveness measure $\chi$ is defined as $\chi(w) = |\sre(w\$)|$, where $\$$ is a unique terminal symbol that does not occur elsewhere in $w$.
By the above definition, $\chi(w)$ is exactly the size of a smallest suffixient set for $w\$$.
It is known that $\chi(w) = |\sre(w\$)| \leq |\sre(w)| + 2$ holds for every string $w$~\cite{DBLP:conf/spire/NavarroRU25}.
Figure~\ref{fig:suffixient_sre_ses} illustrates a smallest suffixient set and the corresponding super-maximal right extensions.

A \emph{bidirectional macro scheme} (BMS) for a string $w$ is a representation that partitions $w$ into consecutive phrases $f_1 \cdots f_k$.
Each phrase $f_i$ is either stored explicitly as a character or represented by a pointer $(p,\ell)$ meaning that $f_i$ equals the substring $w[p..p+\ell-1]$.
Such a scheme induces a transition function $\varphi:[n]\to [n]\cup\{\bot\}$ on text positions.
If position $i$ is stored explicitly, then $\varphi(i)=\bot$.
If position $i$ belongs to a copied phrase with pointer $(p,\ell)$ and is its $t$-th position, then $\varphi(i) = p + t - 1$.
We call the scheme \emph{valid} if for every position $i\in[n]$ there exists $k\ge 0$ such that $\varphi^k(i)=\bot$.
Equivalently, the scheme is valid if the reference relation contains no directed cycle.
We denote by $b(w)$ the minimum number of phrases in a valid bidirectional macro scheme for $w$.

\section{Substring Equation Systems}
A \emph{substring equation system} for a string $w$ of length $n$ consists of substring-equality constraints of the form $w[i..i+\ell-1] = w[j..j+\ell-1]$ and character-assignment constraints of the form $w[k] = c$.
We formalize substring equation systems as follows.
\begin{definition}[Substring equation system (SES)]
An instance of a \emph{substring equation system} for a string $w\in\Sigma^n$ is a triple $(n,\Eq,\Ch)$, where $n$ is the length of $w$, $\Eq$ is a finite set of substring-equality constraints, and $\Ch$ is a finite set of character-assignment constraints.
Each element of $\Eq$ is a pair $(\{i, j\}, \ell)$ with $1 \leq i < j \leq n$ and $1 \leq \ell \leq n - j + 1$, representing the equation $w[i..i+\ell-1] = w[j..j+\ell-1]$.
Each element of $\Ch$ is a pair $(k,c)$ with $1\le k\le n$ and $c\in\Sigma$, representing the equation $w[k]=c$.
We say that $(n,\Eq,\Ch)$ \emph{represents} $w$ if $w$ satisfies all constraints in $\Eq$ and $\Ch$ and $w$ is the unique string in $\Sigma^n$ satisfying them.
The \emph{size} of the system is defined as $|\Eq|+|\Ch|$.
\end{definition}
There exists an $O(n + |\Eq| + |\Ch|)$-time algorithm that decides satisfiability and uniqueness of a substring equation system and, if the system is satisfiable, outputs a string $w\in\Sigma^n$ satisfying all constraints~\cite{DBLP:journals/tcs/GawrychowskiKRR20}.
When the satisfying string is unique, the system can be viewed as a compact representation of the string.
Analogously to other compression-based repetitiveness measures, we define $s(w)$ as the minimum size of an SES that represents $w$.

SES can be seen as a flexible way of expressing undirected copy-paste constraints between substrings.
In this sense, it generalizes classical directed copy-paste representations such as BMSs.
Unlike BMSs, an SES does not require an explicit partitioning of the string into phrases.
In particular, every valid BMS can be converted into an SES of the same size.
\begin{theorem} \label{thm:bms_to_ses}
For every string $w\in\Sigma^n$ and every valid bidirectional macro scheme for $w$ with $k$ phrases,
there exists a substring equation system (SES) of size $k$ that represents $w$.
\end{theorem}
\begin{proof}
Let $f_1\cdots f_k$ be a valid bidirectional macro scheme for $w$.
For each phrase $f_i$, let $[a_i, b_i]$ be its interval in $w$.
If $f_i$ is stored explicitly as a character $c$, then we add the character constraint $(a_i,c)$ to $\Ch$.
Otherwise $f_i$ is represented by a pointer $(p,\ell)$ with $\ell=b_i-a_i+1$, meaning that
$w[a_i..b_i]=w[p..p+\ell-1]$; in this case we add the substring-equality constraint $(\{a_i,p\},\ell)$ to $\Eq$.
Let $(n,\Eq,\Ch)$ be the resulting SES.
By construction, $w$ satisfies all constraints and 
the size of the SES is $|\Eq|+|\Ch|=k$.

It remains to show that the string satisfying all constraints is unique.
Consider the transition function $\varphi$ induced by the macro scheme.
Every substring-equality constraint $(\{a_t,p\},\ell)$ implies that for each $0\leq q<\ell$ we have $\varphi(a_t+q)=p+q$,
while every character constraint $(a_t,c)$ corresponds to $\varphi(a_t)=\bot$.
Because the macro scheme is valid, the directed graph on positions induced by $\varphi$ is acyclic and every position reaches some position with $\varphi=\bot$.
Take a topological order of the graph obtained by reversing the direction of every reference edge.
In this order, each position either has an explicit character constraint fixing its value,
or it is constrained to be equal to a position that appears earlier in the order.
Thus, by induction along the order, the character at every position is uniquely determined.
Hence, $(n,\Eq,\Ch)$ represents $w$.
\end{proof}
The inequality $s(w)\le b(w)$ follows immediately from \Cref{thm:bms_to_ses}.

A set $\Eq$ of substring-equality constraints of an SES naturally induces an undirected \emph{referencing graph} $H_{\Eq}$ on the text positions.
Each equation $(\{i, j\}, \ell)$ adds the edge $\{i + h, j + h\}$ for every $0 \leq h < \ell$.
The vertex sets of the connected components of $H_{\Eq}$ form a partition of the text positions $[n]$.
We denote this partition by $\mathcal{C}(H_\Eq)$ and define $\kappa(\Eq) = |\mathcal{C}(H_\Eq)|$.
For any set $\Eq$ of substring equations satisfied by a string $w$, exactly $\kappa(\Eq)$ character assignments are necessary and sufficient for $(|w|, \Eq, \Ch)$ to represent $w$.
Indeed, the equations force all positions in each component to contain the same character, and every component must contain a character assignment that determines this character.
Conversely, assigning the character of $w$ to one position in each component uniquely determines the characters at all positions.

A substring equation satisfied by $w$ is \emph{maximal} if its two occurrences cannot be extended to the left or to the right while preserving their equality.
The following observation restricts the substring equations that need to be considered.
\begin{observation}\label{obs:maximal_ses}
For every string $w$, there exists a minimum-size SES representing $w$ in which every substring equation has length at least $2$ and is maximal.
\end{observation}
Indeed, removing a length-$1$ equation increases the number of connected components by at most one, so at most one additional character assignment is needed to keep every position uniquely determined.
Moreover, replacing a non-maximal equation by a maximal extension does not remove any edge from the referencing graph and therefore cannot increase the number of connected components.

 \section{Constructing an SES from Suffixient Sets}
In this section, we show how to construct an SES of size $O(\chi(w_{\mathrm{init}}))$ for an input string $w_{\mathrm{init}}$.
We assume that $w_{\mathrm{init}}$ is nonempty and define $w = w_{\mathrm{init}}\$$ by appending a unique terminal symbol, so that $|\sre(w)| = \chi(w_{\mathrm{init}})$.
An SES for $w_{\mathrm{init}}$ can be obtained from an SES for $w$ by removing the character-assignment constraint for $\$$ and reducing the length by one, without increasing its size.
We first construct an SES of size $O(\chi(w_{\mathrm{init}}))$ and then prove that $w$ is its unique solution.

Let $\mathcal{E}$ be an arbitrary smallest suffixient set for $w$, so $|\mathcal{E}| = \chi(w_{\mathrm{init}})$.
We number the elements of $\mathcal{E}$ by the colexicographic order of $w[1..e_i-1]$, so that $w[1..e_i-1]^R \prec w[1..e_{i+1}-1]^R$ for $1 \leq i < \chi(w_{\mathrm{init}})$.
Let $\ell_i$ be the length of the longest common suffix of $w[1..e_i-1]$ and $w[1..e_{i+1}-1]$ for $1 \leq i < \chi(w_{\mathrm{init}})$.

We define the SES $\mathcal{S} = (n, \Eq, \Ch)$ so that $n = |w|$, $\Eq = \{(\{e_i-\ell_i, e_{i+1}-\ell_i\}, \ell_i) \mid 1 \leq i < \chi(w_{\mathrm{init}}),\ \ell_i > 0\}$,
and $\Ch$ consists of exactly one pair $(p, c)$ with $w[p] = c$ for each distinct character $c$ in $w$.
By definition, each substring-equality constraint corresponds to the longest common suffix of $w[1..e_i-1]$ and $w[1..e_{i+1}-1]$.
Since $|\Eq| \leq \chi(w_{\mathrm{init}}) - 1$ and $|\Ch| = \sigma(w) \leq \chi(w_{\mathrm{init}})$, this SES has at most $2\chi(w_{\mathrm{init}}) - 1$ constraints.
Figure~\ref{fig:suffixient_sre_ses} illustrates a smallest suffixient set and the SES constructed from it.

\begin{figure}[th]
  \centering
  \includegraphics[width=0.98\textwidth]{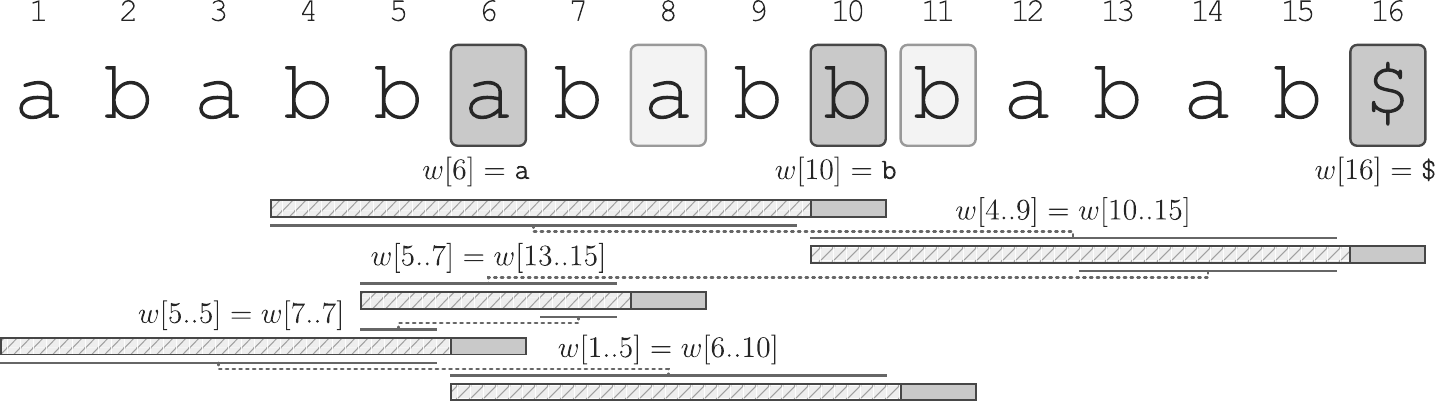}
  \caption{
    An example of an SES constructed from a smallest suffixient set for $w = \texttt{ababbababbbabab\$}$.
    The rounded boxes mark the positions in the smallest suffixient set,
    and the darker boxes correspond to the character-assignment constraints.
    The blocks below $w$ show the super-maximal right extensions $xc$ in the colexicographic order of $w[1..e_i-1]$, with $(e_1, \dots, e_5) = (10, 16, 8, 6, 11)$ from top to bottom.
    The hatched segment of each block marks $x$, and the solid gray segment marks the final character $c$.
    Each pair of solid segments connected by a dotted line represents the substring-equality constraint for the two adjacent blocks.
  }
  \label{fig:suffixient_sre_ses}
\end{figure}

To show that $w$ is the unique solution to $\mathcal{S}$, we first prove the following lemma.
\begin{lemma} \label{lem:suffixient_lcs}
Let $\ell_{i,j}$ be the length of the longest common suffix of $w[1..e_i-1]$ and $w[1..e_j-1]$ with $i < j$.
Every solution $w'$ to $\mathcal{S}$ satisfies
$w'[e_i-\ell_{i,j}..e_i-1] = w'[e_j-\ell_{i,j}..e_j-1]$.
\end{lemma}
\begin{proof}
The claim is immediate when $\ell_{i,j} = 0$.
Let $w'$ be any solution to $\mathcal{S}$ and let $L = \ell_{i,j} > 0$.
By the definition of the order on $\mathcal{E}$, $w[1..e_i-1]^R \prec w[1..e_k-1]^R \prec w[1..e_j-1]^R$ for every $i < k < j$.
Since $w[1..e_i-1]^R$ and $w[1..e_j-1]^R$ share a prefix of length $L$, every $w[1..e_k-1]^R$ with $i \leq k \leq j$ has the same prefix.
Hence, $\ell_k \geq L$ for every $i \leq k < j$.
For each such $k$, the constraint in $\Eq$ for $e_k$ and $e_{k+1}$ implies $w'[e_k - L..e_k-1] = w'[e_{k+1} - L..e_{k+1}-1]$.
Thus, $w'[e_i - L..e_i-1] = w'[e_j - L..e_j-1]$ holds by transitivity.
\end{proof}

We now prove the main theorem of this section, which implies $s(w_{\mathrm{init}}) \leq \chi(w_{\mathrm{init}}) + \sigma(w_{\mathrm{init}}) - 1 \leq 2\chi(w_{\mathrm{init}})$.
\begin{theorem}
The string $w$ is the unique solution to the SES $\mathcal{S}$.
\end{theorem}
\begin{proof}
Let $w'$ be any solution to $\mathcal{S}$.
We prove that $w'[i-\ell..i-1] = w'[j-\ell..j-1]$ for any two occurrences $w[i-\ell..i-1]$ and $w[j-\ell..j-1]$ of the same substring in $w$ with $i < j$.
We proceed by induction on $\ell$ in decreasing order.
The case $\ell = n$ is immediate because there are no two distinct occurrences of length $n$.

Assume that the claim holds for all pairs of occurrences of the same substring of length at least $\ell+1$.
Let $i$ and $j$ be positions such that $w[i-\ell..i-1] = w[j-\ell..j-1] = u$ with $i < j$.
Since the terminal symbol occurs only once in $w$, we have $i < j \leq n$.
Let $c = w[i]$ and $c' = w[j]$.
If $c = c'$, we have $w[i-\ell..i] = w[j-\ell..j] = uc = uc'$.
Since these occurrences have length $\ell+1$, the induction hypothesis gives $w'[i-\ell..i] = w'[j-\ell..j]$ and hence $w'[i-\ell..i-1] = w'[j-\ell..j-1]$.

Now suppose that $c \neq c'$.
Then, $uc, uc' \in \re(w)$.
By \Cref{def:smallest_suffixient_set}, $\mathcal{E}$ contains the ending position of an occurrence of each super-maximal right extension.
Since every right extension is a suffix of a super-maximal right extension, $uc$ and $uc'$ have occurrences ending at positions in $\mathcal{E}$.
Let $i', j' \in \mathcal{E}$ be positions such that $w[1..i']$ and $w[1..j']$ have $uc$ and $uc'$ as suffixes, respectively.
By the induction hypothesis for the occurrences of $uc$ and $uc'$, we have $w'[i-\ell..i] = w'[i'-\ell..i']$ and $w'[j-\ell..j] = w'[j'-\ell..j']$.
The prefixes $w[1..i'-1]$ and $w[1..j'-1]$ both have $u$ as a suffix, so \Cref{lem:suffixient_lcs} and these equalities imply that
(i) $w'[1..i-1]$ and $w'[1..i'-1]$ have a common suffix of length $\ell$,
(ii) $w'[1..i'-1]$ and $w'[1..j'-1]$ have a common suffix of length $\ell$, and
(iii) $w'[1..j-1]$ and $w'[1..j'-1]$ have a common suffix of length $\ell$.
By transitivity, $w'[i-\ell..i-1] = w'[j-\ell..j-1]$.
Thus, $w'[i-\ell..i-1] = w'[j-\ell..j-1]$ holds for any two occurrences $w[i-\ell..i-1]$ and $w[j-\ell..j-1]$ of the same substring in $w$.

For $\ell = 1$, we obtain $w'[p] = w'[q]$ whenever $w[p] = w[q]$.
Since $\Ch$ contains one constraint for each distinct character in $w$, we have $w'[q] = w[q]$ for every $q \in [n]$.
Therefore, $w' = w$.
The string $w$ satisfies all constraints in $\mathcal{S}$, so it is the unique solution to $\mathcal{S}$.
\end{proof}

Finally, we combine the preceding results with a construction in linear time to obtain the following theorem.
\begin{theorem}\label{thm:ses_linear_construction}
Given a nonempty string $w_{\mathrm{init}}$ over a linearly sortable alphabet, an SES of size $O(\chi(w_{\mathrm{init}}))$ representing $w_{\mathrm{init}}$ can be constructed in $O(|w_{\mathrm{init}}|)$ time.
\end{theorem}
\begin{proof}
The SES constructed above represents $w = w_{\mathrm{init}}\$$ and has $O(\chi(w_{\mathrm{init}}))$ constraints.
It remains to show that this SES can be constructed in $O(n)$ time, where $n = |w| = |w_{\mathrm{init}}| + 1$.

A smallest suffixient set $\mathcal{E}$ can be computed in $O(n)$ time~\cite{DBLP:conf/spire/CenzatoOP24}.
To order $\mathcal{E}$ and determine each $\ell_i$, we use the suffix array and the LCP array $\mathrm{LCP}$ of $w^R$~\cite{DBLP:journals/siamcomp/ManberM93}, where $\mathrm{LCP}[r]$ is the length of the longest common prefix of the $r$-th and $(r+1)$-th suffixes in lexicographic order.
Both arrays can be constructed in $O(n)$ time~\cite{DBLP:conf/focs/Farach97,DBLP:conf/cpm/KasaiLAAP01}.
Since the reversals of prefixes of $w$ are suffixes of $w^R$, we can obtain the order on $\mathcal{E}$ by scanning the suffix array.
Let $w[1..e_i-1]^R$ be the $r_i$-th suffix of $w^R$ in lexicographic order.
Then, $\ell_i = \min\{\mathrm{LCP}[r_i], \dots, \mathrm{LCP}[r_{i+1}-1]\}$ for $1 \leq i < \chi(w_{\mathrm{init}})$.
All values $\ell_i$ can be computed in $O(n)$ time by a single scan of the LCP array.
The set $\Ch$ can be constructed in $O(n)$ time by sorting the positions of $w$ by their characters and retaining one position per character, so the total construction time is $O(n)$.
By removing the character-assignment constraint for $\$$ from the SES for $w$ and reducing the length by one, we obtain an SES for $w_{\mathrm{init}}$ in $O(|w_{\mathrm{init}}|)$ total time.
\end{proof}

\section{NP-Hardness of Computing the Minimum SES Size}\label{se:np_hardness}

In this section we prove that computing $s(w)$ is NP-hard by a reduction from \prob{minimum vertex cover} on $4$-regular graphs.

A graph $G = ( V, E )$ is \emph{simple} if it has no self-loops and at most one edge between any pair of distinct vertices.
$G$ is \emph{$B$-regular} if every vertex has degree $B$.
A \emph{vertex cover} of $G$ is a set $C \subseteq V$ such that every edge in $E$ has at least one endpoint in $C$.
We denote the minimum cardinality of a vertex cover of $G$ by $\tau(G)$.
The \prob{minimum vertex cover} problem is to compute $\tau(G)$ for a given graph $G$.
We use the following hardness result.
\begin{theorem}[{\cite[Corollary~18]{DBLP:journals/tcs/ChlebikC06}}]\label{thm:vertex_cover_hardness}
It is NP-hard to approximate \prob{minimum vertex cover} within a factor of $53/52$, even when the input is restricted to simple $4$-regular graphs.
\end{theorem}

Let $G = (V, E)$ be a simple $4$-regular graph with $n = |V|$ vertices and $m = |E| = 2n$ edges.
We identify $V$ with $[n]$.
For each vertex $v \in V$, fix an arbitrary ordering of the four edges incident to $v$ and denote them by $e_1(v), e_2(v), e_3(v)$, and $e_4(v)$.
We define the alphabet $\Sigma_G = \{\S{v}, \T{v} \mid v \in V\} \cup \{\$_{v,i}, \#_{v,i} \mid v \in V, 1 \leq i \leq 4\} \cup \{\X{e}, \A{e}, \B{e} \mid e \in E\}$.
The symbols $\S{v}$ and $\T{v}$ mark the beginnings of gadget strings, while $\A{e}$ and $\B{e}$ mark their ends.
These boundary markers control the possible substring equalities.
The symbol $\X{e}$ appears in the gadgets associated with both endpoints of $e$.
Connecting its occurrences will encode that $e$ is covered.
The symbols $\$_{v,i}$ and $\#_{v,i}$ are separators.
All separator symbols are distinct, and each occurs exactly once in the string constructed below.
The alphabet has size $|\Sigma_G| = 10n + 3m = 16n$.
Over this alphabet, we define gadget strings $s_{v,i}$ and $t_{v,i}$ for every $v \in V$ and $1 \leq i \leq 4$.
We then construct two intermediate strings $S_G$ and $T_G$ and the final string $w_G$ as follows.
\begin{align*}
    s_{v,i} &= \S{v}\X{e_1(v)} \cdots \X{e_i(v)}\A{e_i(v)}, \\
    t_{v,i} &= \T{v}\X{e_1(v)} \cdots \X{e_i(v)}\B{e_i(v)}, \\
    S_G &=
    \bigl(s_{1,1}\$_{1,1} s_{1,2}\$_{1,2} \cdots s_{1,4}\$_{1,4}\bigr)
    \cdots
    \bigl(s_{n,1}\$_{n,1} s_{n,2}\$_{n,2} \cdots s_{n,4}\$_{n,4}\bigr), \\
    T_G &=
    \bigl(t_{1,1}\#_{1,1} t_{1,2}\#_{1,2} \cdots t_{1,4}\#_{1,4}\bigr)
    \cdots
    \bigl(t_{n,1}\#_{n,1} t_{n,2}\#_{n,2} \cdots t_{n,4}\#_{n,4}\bigr), \\
    w_G &= S_G T_G.
\end{align*}
For every $v \in V$ and $1 \leq i \leq 4$, both $s_{v,i}$ and $t_{v,i}$ have length $i + 2$.
Hence the four gadget strings associated with one vertex have total length $\sum_{i=1}^4 (i + 2) = 18$ in each of $S_G$ and $T_G$.
Each intermediate string also contains four separator symbols per vertex.
It follows that $|S_G| = |T_G| = (18 + 4)n = 22n$ and $|w_G| = 44n$.
Figure~\ref{fig:np_hardness_gadgets} shows an example of the gadget strings for two adjacent vertices.

\begin{figure}[t]
    \centering
    \begin{minipage}[c]{0.52\linewidth}
        \setlength{\mathindent}{0pt}
        \setlength{\jot}{3pt}
        \begin{alignat*}{2}
            s_{1,1} &= \S{1}\X{a}\A{a}, \quad & s_{4,1} &= \S{4}\underline{\X{c}\A{c}}, \\
            s_{1,2} &= \S{1}\X{a}\X{b}\A{b}, \quad & s_{4,2} &= \S{4}\X{c}\X{h}\A{h}, \\
            s_{1,3} &= \S{1}\X{a}\X{b}\underline{\X{c}\A{c}}, \quad & s_{4,3} &= \S{4}\X{c}\X{h}\X{k}\A{k}, \\
            s_{1,4} &= \S{1}\X{a}\X{b}\X{c}\X{d}\A{d}, \quad & s_{4,4} &= \S{4}\X{c}\X{h}\X{k}\X{l}\A{l}, \\[0.5ex]
            t_{1,1} &= \T{1}\X{a}\B{a}, \quad & t_{4,1} &= \T{4}\underline{\X{c}\B{c}}, \\
            t_{1,2} &= \T{1}\X{a}\X{b}\B{b}, \quad & t_{4,2} &= \T{4}\X{c}\X{h}\B{h}, \\
            t_{1,3} &= \T{1}\X{a}\X{b}\underline{\X{c}\B{c}}, \quad & t_{4,3} &= \T{4}\X{c}\X{h}\X{k}\B{k}, \\
            t_{1,4} &= \T{1}\X{a}\X{b}\X{c}\X{d}\B{d}, \quad & t_{4,4} &= \T{4}\X{c}\X{h}\X{k}\X{l}\B{l}.
        \end{alignat*}
    \end{minipage}
    \begin{minipage}[c]{0.46\linewidth}
        \centering
        \includegraphics[width=0.9\linewidth]{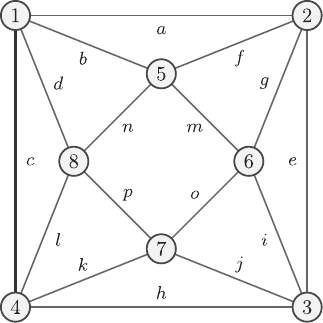}
    \end{minipage}
    \caption{
        A simple $4$-regular graph $G$ and the corresponding gadget strings associated with vertices $1$ and $4$.
        The incident edges are ordered alphabetically at each vertex.
        The underlined pairs of occurrences of $\X{c}\A{c}$ and $\X{c}\B{c}$ give two equations in $\EqA$ and $\EqB$, respectively.
        Both equations correspond to the bold edge $c = \{1, 4\}$.
    }
    \label{fig:np_hardness_gadgets}
\end{figure}

We first state the following lemma which characterizes the maximal equations on $w_G$.
\begin{lemma}\label{lem:maximal_equations}
The maximal substring equations of length at least $2$ between occurrences starting at different positions in $w_G$ form the following five families.
\begin{itemize}
\item $\EqS$ consists of the equations between the occurrences of $\S{v}\X{e_1(v)} \cdots \X{e_i(v)}$ in $s_{v,i}$ and $s_{v,j}$ for every $v \in V$ and $1 \leq i < j \leq 4$.
\item $\EqT$ consists of the equations between the occurrences of $\T{v}\X{e_1(v)} \cdots \X{e_i(v)}$ in $t_{v,i}$ and $t_{v,j}$ for every $v \in V$ and $1 \leq i < j \leq 4$.
\item $\EqA$ consists of the equations between the suffixes $\X{e}\A{e}$ of $s_{u,i}$ and $s_{v,j}$ for every edge $e = \{u, v\} \in E$, where $i$ and $j$ satisfy $e_i(u) = e_j(v) = e$.
\item $\EqB$ consists of the equations between the suffixes $\X{e}\B{e}$ of $t_{u,i}$ and $t_{v,j}$ for every edge $e = \{u, v\} \in E$, where $i$ and $j$ satisfy $e_i(u) = e_j(v) = e$.
\item $\EqX = \bigcup_{v \in V} \EqXv{v}$,\newline where $\EqXv{v}$ consists of the equations between the occurrences of $\X{e_1(v)} \cdots \X{e_{\min\{i,j\}}(v)}$ in $s_{v,i}$ and $t_{v,j}$ for all $2 \leq i, j \leq 4$.
\end{itemize}
\end{lemma}
\begin{proof}
Every equation in these families is maximal.
For an equation in $\EqS$ or $\EqT$, the characters following the common prefix differ, and distinct separators prevent an extension to the left whenever both preceding characters exist.
For an equation in $\EqA$ or $\EqB$, the distinct separators prevent an extension to the right.
The characters preceding $\X{e}$ are either distinct initial symbols or the symbols of two different edges, since another common edge would violate simplicity of $G$.
An equation in $\EqX$ cannot be extended to the left because $\S{v}$ and $\T{v}$ differ, and the characters following the shorter edge-symbol sequence also differ.

We next show that these families contain every possible maximal equation of length at least $2$.
No maximal equation of length at least $2$ crosses a gadget boundary because every separator occurs exactly once in $w_G$.
Thus, both sides of such an equation are contained in gadget strings.
An equation containing $\S{v}$ or $\T{v}$ consists of the longest common prefixes of two gadgets associated with $v$.
It therefore belongs to $\EqS$ or $\EqT$.
Each $\A{e}$ occurs only at the ends of the two gadgets associated with the endpoints of $e$.
The common suffix of these gadgets is exactly $\X{e}\A{e}$ because simplicity of $G$ rules out another edge between the same two vertices.
The same argument applies to $\B{e}$, giving the equations in $\EqA$ and $\EqB$.

It remains to consider equations containing only symbols of the form $\X{e}$.
If a common substring of length at least $2$ occurs in gadgets associated with different vertices $u$ and $v$, its first two symbols correspond to different edges between $u$ and $v$, contradicting simplicity of $G$.
Hence, the two occurrences belong to gadgets associated with the same vertex $v$.
Since the edge symbols are distinct, $\X{e_h(v)}$ is the $h$-th character of the edge-symbol sequence of every gadget associated with $v$ that contains it.
Thus, the two sides of the equation align the same offsets.
If both sides lie in $s$-gadgets, their maximal common substring is one of the prefixes in $\EqS$.
If both sides lie in $t$-gadgets, it is one of the prefixes in $\EqT$.
The remaining case has one side in an $s$-gadget and the other in a $t$-gadget.
Their maximal common substring starts with $\X{e_1(v)}$ and ends with the last edge symbol of the shorter gadget, giving an equation in $\EqXv{v}$.
\end{proof}

Together with \Cref{obs:maximal_ses}, \Cref{lem:maximal_equations} implies that there exists a minimum-size SES for $w_G$ whose set $\Eq$ of substring-equality constraints is contained in $\EqS \cup \EqT \cup \EqA \cup \EqB \cup \EqX$.
For each $v \in V$ and $1 \leq i \leq 3$, let $\EqS^\circ(v, i)$ denote the equation between the occurrences of $\S{v}\X{e_1(v)} \cdots \X{e_i(v)}$ in $s_{v,i}$ and $s_{v,i+1}$.
Define $\EqT^\circ(v, i)$ analogously for $t_{v,i}$ and $t_{v,i+1}$.
We write $\EqS^\circ = \{\EqS^\circ(v, i) \mid v \in V, 1 \leq i \leq 3\}$ and $\EqT^\circ = \{\EqT^\circ(v, i) \mid v \in V, 1 \leq i \leq 3\}$.
For each $v \in V$, let $\EqX^\circ(v)$ denote the equation between the occurrences of $\X{e_1(v)} \cdots \X{e_4(v)}$ in $s_{v,4}$ and $t_{v,4}$.
We write $\EqX^\circ = \{\EqX^\circ(v) \mid v \in V\}$.
For $C \subseteq V$, define
\[
\Eq(C) = \EqS^\circ \cup \EqT^\circ \cup \EqA \cup \EqB \cup \{\EqX^\circ(v) \mid v \in C\}.
\]
The following lemma describes the connected components of the referencing graph $H_{\Eq(\emptyset)}$ induced by $\Eq(\emptyset) = \EqS^\circ \cup \EqT^\circ \cup \EqA \cup \EqB$.
\begin{lemma}\label{lem:eq_empty_components}
The connected components of $H_{\Eq(\emptyset)}$ have the following structure.
\begin{itemize}
\item For every $e \in E$, all occurrences of $\X{e}$ in $S_G$ form one component, and all occurrences of $\X{e}$ in $T_G$ form another component.
These two components are distinct.
\item For every character $c \in \Sigma_G \setminus \{\X{e} \mid e \in E\}$, all occurrences of $c$ form one component.
\end{itemize}
\end{lemma}
\begin{proof}
Every edge of $H_{\Eq(\emptyset)}$ joins two positions containing the same character.
For each $e \in E$, the equations in $\EqS^\circ$ connect the occurrences of $\X{e}$ within the $s$-gadgets associated with each endpoint of $e$.
The corresponding equation in $\EqA$ connects the occurrences associated with the two endpoints.
Thus, all occurrences of $\X{e}$ in $S_G$ form one component.
The same argument with $\EqT^\circ$ and $\EqB$ gives one component containing all occurrences of $\X{e}$ in $T_G$.
No equation in $\Eq(\emptyset)$ has one side in $S_G$ and the other in $T_G$, so these two components are distinct.

For each $v \in V$, the equations in $\EqS^\circ$ connect all occurrences of $\S{v}$, and the equations in $\EqT^\circ$ connect all occurrences of $\T{v}$.
For each $e \in E$, the corresponding equation in $\EqA$ connects the two occurrences of $\A{e}$, and the corresponding equation in $\EqB$ connects the two occurrences of $\B{e}$.
Every separator occurs exactly once and therefore forms an isolated component.
Thus, all occurrences of each character except $\X{e}$ for $e \in E$ form one connected component.
\end{proof}

We next use this component structure to give a more restricted normal form for the substring equations in a minimum-size SES.
\begin{lemma}\label{lem:ses_normal_form}
There exist a minimum-size SES $(|w_G|, \Eq, \Ch)$ representing $w_G$ and a set $C \subseteq V$ such that
$\Eq = \EqS^\circ \cup \EqT^\circ \cup \EqA \cup \EqB \cup \{\EqX^\circ(v) \mid v \in C\}$.
\end{lemma}
\begin{proof}
\Cref{obs:maximal_ses} and \Cref{lem:maximal_equations} allow us to restrict the substring equations to a subset $\mathsf{F}$ of $\EqS \cup \EqT \cup \EqA \cup \EqB \cup \EqX$.
For every $\mathsf{F}$, the minimum number of character assignments required for an SES with substring-equality constraints $\mathsf{F}$ to represent $w_G$ is $\kappa(\mathsf{F})$.
Hence, $s(w_G)$ is the minimum of $|\mathsf{F}| + \kappa(\mathsf{F})$ over all such subsets $\mathsf{F}$.

Fix $v \in V$ and consider the four occurrences of $\S{v}$ in $s_{v,1}, \ldots, s_{v,4}$.
An equation associated with $v$ in $\EqS$ connects two of these occurrences, and no equation in another family connects them.
Therefore, connecting all four occurrences requires at least three equations.
The equations $\EqS^\circ(v, 1), \EqS^\circ(v, 2)$, and $\EqS^\circ(v, 3)$ attain this bound and induce the same connected components as all equations in $\EqS$ associated with $v$.
Indeed, all four occurrences of $\S{v}$ are connected by these equations,
and these equations also connect all occurrences of $\X{e_h(v)}$ in $s_{v,h}, \ldots, s_{v,4}$ for $1 \leq h \leq 4$.
If $\mathsf{F}$ contains only $r < 3$ equations from $\EqS$ associated with $v$, then $\kappa(\mathsf{F})$ exceeds its value for the three equations above by at least $3-r$.
The three equations above already connect all positions connected by the equations in $\EqS$ for $v$, so adding more of these equations does not change $\kappa(\mathsf{F})$.
The same argument applies to the occurrences of $\T{v}$.
Therefore, $\EqS^\circ \cup \EqT^\circ$ is an optimal choice among the equations in $\EqS \cup \EqT$.

For each $e \in E$, the two occurrences of $\A{e}$ can be connected only by the corresponding equation in $\EqA$.
Choosing this equation increases $|\mathsf{F}|$ by one and decreases $\kappa(\mathsf{F})$ by at least one, so it does not increase $|\mathsf{F}| + \kappa(\mathsf{F})$.
The same argument applies to the two occurrences of $\B{e}$ and the corresponding equation in $\EqB$.
Thus, an optimal choice may include all equations in $\EqA \cup \EqB$.

By \Cref{lem:eq_empty_components}, the occurrences of each $\X{e}$ belong to exactly two connected components in $H_{\Eq(\emptyset)}$, one in $S_G$ and one in $T_G$.
For each $v \in V$, the equation $\EqX^\circ(v)$ between $s_{v,4}$ and $t_{v,4}$ merges these two components for each of $\X{e_1(v)}, \ldots, \X{e_4(v)}$.
Every equation in $\EqXv{v}$ is between two occurrences of a prefix of $\X{e_1(v)} \cdots \X{e_4(v)}$, so adding these equations does not change the connected components.
Thus, choosing $\EqX^\circ(v)$ alone induces the same connectivity among text positions as choosing all equations in $\EqXv{v}$.
It follows that an optimal choice among the equations in $\EqX$ has the form $\{\EqX^\circ(v) \mid v \in C\}$ for some $C \subseteq V$.
Combining these choices gives an optimal $\mathsf{F}$ of the form stated in the lemma.
The $\kappa(\mathsf{F})$ character assignments described above then give the required minimum-size SES.
\end{proof}

\Cref{lem:ses_normal_form} shows that the substring equations of a minimum-size SES can be chosen in this form.
The component structure in \Cref{lem:eq_empty_components} gives an exact expression for the minimum size of an SES representing $w_G$.
\begin{theorem}\label{thm:ses_vertex_cover_relation}
For every simple $4$-regular graph $G$ with $n$ vertices, the string $w_G$ satisfies $s(w_G) = 26n + \tau(G)$.
\end{theorem}
\begin{proof}
For $C \subseteq V$, let $\lambda(C)$ be the number of edges with no endpoint in $C$.
The four fixed families $\EqS^\circ$, $\EqT^\circ$, $\EqA$, and $\EqB$ in $\Eq(C)$ contain $3n$, $3n$, $m$, and $m$ equations, respectively.
Since $m = 2n$, we have $|\Eq(C)| = 10n + |C|$.

\Cref{lem:eq_empty_components} shows that every character other than the edge symbols has one component in $H_{\Eq(C)}$.
For every edge $e = \{u, v\}$, an equation $\EqX^\circ(x)$ merges the $S_G$ and $T_G$ components of $\X{e}$ if and only if $x \in \{u, v\}$.
Therefore, the occurrences of $\X{e}$ belong to one component if $e$ has an endpoint in $C$ and to two components otherwise.
Since $|\Sigma_G| = 16n$, it follows that $\kappa(\Eq(C)) = 16n + \lambda(C)$.
By \Cref{lem:ses_normal_form}, there exists a minimum-size SES whose set of substring equations is $\Eq(C)$ for some $C \subseteq V$.
For each $C$, the minimum number of character assignments is $\kappa(\Eq(C))$.
Therefore,
\begin{align*}
s(w_G)
&= \min_{C \subseteq V} \bigl(|\Eq(C)| + \kappa(\Eq(C))\bigr) \\
&= \min_{C \subseteq V} \bigl((10n + |C|) + (16n + \lambda(C)) \bigr) \\
&= 26n + \min_{C \subseteq V} \bigl(|C| + \lambda(C)\bigr).
\end{align*}

For every $C \subseteq V$, adding one endpoint of each edge counted by $\lambda(C)$ produces a vertex cover of size at most $|C| + \lambda(C)$.
Hence, $\tau(G) \leq |C| + \lambda(C)$.
Conversely, a minimum vertex cover $C$ satisfies $|C| = \tau(G)$ and $\lambda(C) = 0$.
Thus, $\min_{C \subseteq V} (|C| + \lambda(C)) = \tau(G)$, which proves the theorem.
\end{proof}

We now obtain the computational hardness of the minimum SES size.
\begin{theorem}\label{thm:ses_np_hardness}
Computing $s(w)$ is NP-hard.
Moreover, it is NP-hard to approximate $s(w)$ within a factor of $2757/2756$.
\end{theorem}
\begin{proof}
The string $w_G$ has length $44n$ and can be constructed from $G$ in polynomial time.
\Cref{thm:ses_vertex_cover_relation} gives $\tau(G) = s(w_G) - 26n$.
Thus, the mapping $G \mapsto w_G$, followed by subtracting $26n$ from $s(w_G)$, is a polynomial-time reduction from \prob{minimum vertex cover} to computing $s(w)$.
This proves the first claim.

For the approximation claim, let $\rho = 2757/2756$ and let $\alpha$ be a value satisfying $s(w_G) \leq \alpha \leq \rho s(w_G)$.
Define $\widehat{\tau} = \alpha - 26n$.
\Cref{thm:ses_vertex_cover_relation} implies $\widehat{\tau} \geq \tau(G)$.
Every vertex in $G$ is incident to four edges, while $G$ has $2n$ edges.
Consequently, every vertex cover has size at least $n/2$, and hence $n \leq 2\tau(G)$.
Combining these inequalities, we obtain
\[
\widehat{\tau}
= \alpha - 26n
\leq \rho\bigl(26n + \tau(G)\bigr) - 26n
= \rho\tau(G) + 26(\rho - 1)n
\leq \bigl(\rho + 52(\rho - 1)\bigr)\tau(G)
= \frac{53}{52}\tau(G).
\]
Hence, the mappings $G \mapsto w_G$ and $\alpha \mapsto \alpha - 26n$ transform any value within a factor of $\rho$ of $s(w_G)$ into a value within a factor of $53/52$ of $\tau(G)$.
This gives a polynomial-time reduction from approximating \prob{minimum vertex cover} within a factor of $53/52$ on simple $4$-regular graphs to approximating $s(w)$ within a factor of $2757/2756$.
The second claim follows from this reduction and the existing hardness result stated in \Cref{thm:vertex_cover_hardness}.
\end{proof}
 
\section{Constant-Factor Equivalence Between SES and BMS}\label{se:bms_ses_equivalence}

As we have shown, every BMS can be converted into an SES of the same size, so $s(w) \leq b(w)$ for every string $w$.
We now ask whether a converse bound holds up to a constant factor.

In this section, we prove the following theorem.
\begin{theorem} \label{thm:bms_ses_equivalence}
Given an SES $(n, \Eq, \Ch)$ representing a nonempty string $w$, we can construct a BMS of size at most $4(|\Eq| + |\Ch|)$.
\end{theorem}
The above theorem directly implies $b(w) \leq 4s(w)$.
Combining this bound with $s(w) \leq b(w)$ by \Cref{thm:bms_to_ses}, we obtain $s(w) \leq b(w) \leq 4s(w)$ and thus $s(w) \in \Theta(b(w))$.

An SES $(n, \Eq, \Ch)$ representing a string is a \emph{forest SES} if $H_{\Eq}$ forms a forest.
We distinguish edges of $H_{\Eq}$ induced by different constraints, so an SES is not a forest SES if two edges induced by different constraints have the same endpoints.
We assume that there are no two constraints such that every edge induced by one has the same endpoints as an edge induced by the other.
The following lemma shows that every SES can be converted to a forest SES without increasing its size.

\begin{figure}[th]
  \centering
  \includegraphics[width=0.98\textwidth]{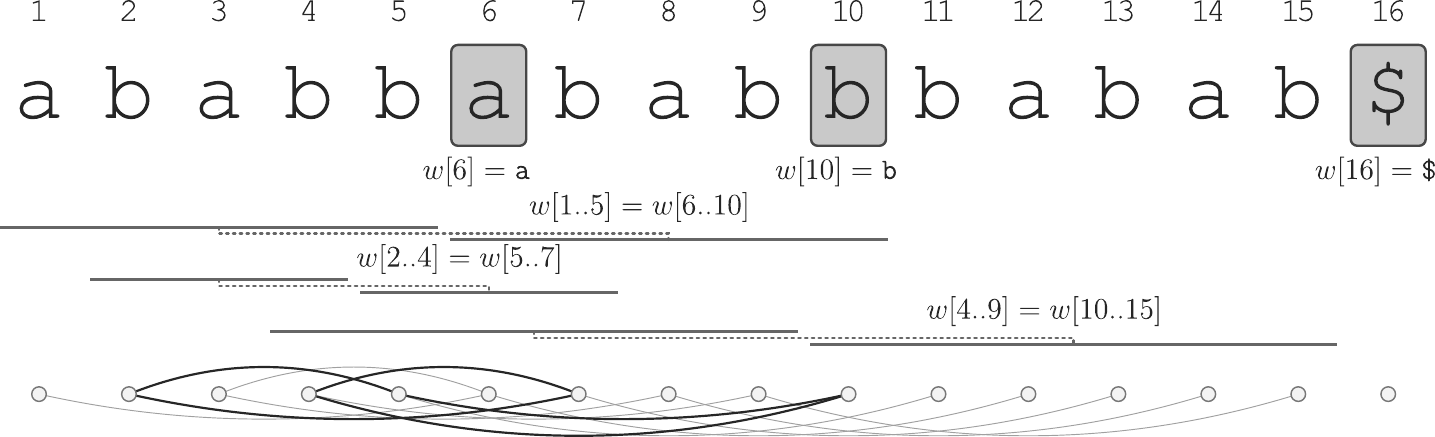}
  \caption{
    An illustration of a non-forest SES representing $w = \texttt{ababbababbbabab\$}$.
Each shaded box marks a character-assignment constraint.
    Each pair of solid segments connected by a dotted line represents the substring-equality constraint written beside it.
    The lower diagram shows the referencing graph $H_{\Eq}$, and bold edges mark a cycle with $o_{\min} = 0$.
    In this example, replacing $w[2..4] = w[5..7]$ with $w[3..4] = w[6..7]$ removes the cycle from $H_{\Eq}$.
  }
  \label{fig:cycle_example_nonforest_ses}
\end{figure}

\begin{lemma} \label{lem:forest_ses}
Given an SES representing a string $w$, we can construct a forest SES representing $w$ without increasing its size.
\end{lemma}
\begin{proof}
We show a method for converting an SES $(n, \Eq, \Ch)$ into a forest SES without increasing its size.
We write $\Eq = \{e_k\}_{k=1}^{m}$, where $m = |\Eq|$ and $e_k = (\{i_k, j_k\}, \ell_k)$ for $1 \leq k \leq m$.

Assume that the given SES is not a forest SES.
Then, $H_{\Eq}$ contains a cycle $P = (p_0, p_1, \ldots, p_{|P| - 1}, p_{|P|} = p_0)$.
For each $0 \leq t < |P|$, the edge of $P$ from $p_t$ to $p_{t+1}$ is induced by a constraint $e_{k_t}$ at an offset $o_t$ with $0 \leq o_t < \ell_{k_t}$.
That is, $\{i_{k_t} + o_t, j_{k_t} + o_t\} = \{p_t, p_{t+1}\}$.

Let $o_{\min} = \min_{0 \leq t < |P|} o_t$.
If $o_{\min} > 0$, each selected edge is induced by its corresponding constraint at a positive offset.
Thus, the sequence of vertices $(p'_0, \ldots, p'_{|P|})$ defined by $p'_i = p_i - 1$ forms another cycle in $H_{\Eq}$.
The edges on this cycle are induced by the same constraints as before, with offsets $o_t - 1$ for $0 \leq t < |P|$.
By repeating this shift, we may assume that $o_{\min} = 0$.

If $o_{\min} = 0$, there exists an integer $t$ with $0 \leq t < |P|$ such that $o_t = 0$.
If $\ell_{k_t} > 1$, we replace $e_{k_t} = (\{i_{k_t}, j_{k_t}\}, \ell_{k_t})$ with $(\{i_{k_t} + 1, j_{k_t} + 1\}, \ell_{k_t} - 1)$.
If $\ell_{k_t} = 1$, we simply remove $e_{k_t}$.
This modification removes only the edge $\{i_{k_t}, j_{k_t}\}$ from $H_{\Eq}$ and all other edges remain unchanged~(see \Cref{fig:cycle_example_nonforest_ses}).
Since the removed edge $\{i_{k_t}, j_{k_t}\} = \{p_t, p_{t+1}\}$ lay on the cycle, the connected components of $H_{\Eq}$ remain unchanged, and the resulting SES represents the same string.
The modification also does not increase the size of the SES.

Each modification decreases the sum of the lengths of the constraints in $\Eq$.
Thus, by repeating this procedure while $H_{\Eq}$ contains a cycle, we can obtain a forest SES representing the same string without increasing its size.
\end{proof}

For the remainder of this section, we assume that the given SES $(n, \Eq, \Ch)$ is a forest SES.
In what follows, we order the substring-equality constraints arbitrarily and write $\Eq = \{ e_1, \dots, e_m\}$, where $m = |\Eq|$ and $e_k = (\{i_k, j_k\}, \ell_k)$ for $1 \leq k \leq m$.

We next introduce the notation for permutations and transpositions used in this section.
We represent a permutation of $[n]$ by a bijection $\pi: [n] \to [n]$.
For any $1 \leq i, j \leq n$, let $\tau_{i,j}: [n] \to [n]$ 
be the transposition defined by $\tau_{i,j}(i) = j$, $\tau_{i,j}(j) = i$, and $\tau_{i, j}(k) = k$ for every $k \in [n] \setminus \{ i, j\}$.
For any two bijections $\pi$ and $\rho$ on $[n]$, $\pi \circ \rho$ denotes their composition, which satisfies $(\pi \circ \rho)(k) = \pi(\rho(k))$ for every $k \in [n]$.

For a position $i \in [n]$, consider the sequence
$i_0, i_1, \ldots$ defined by
$i_0 = i$ and $i_{t+1} = \pi(i_t)$.
Because $\pi$ is a bijection on the finite set $[n]$,
this sequence eventually returns to $i$.
Let $r$ be the smallest positive integer such that $i_r = i$.
The \emph{orbit} of $i$ under $\pi$ is $O_\pi(i) = \{i_0, \ldots, i_{r-1}\}$.
We denote the set of all orbits of $\pi$ by
$\mathcal{O}(\pi) = \{O_\pi(i) \mid i \in [n]\}$.
For a permutation $\pi$ and positions $i, j \in [n]$ belonging to different orbits of $\pi$, composing $\pi$ with $\tau_{i,j}$ merges $O_\pi(i)$ and $O_\pi(j)$ into the single orbit $O_{\pi \circ \tau_{i,j}}(i) = O_\pi(i) \cup O_\pi(j)$.
Every other orbit of $\pi$ remains unchanged.

We now use this property to construct the permutation $\pi_\Eq$ whose orbits correspond to the connected components of $H_{\Eq}$.
The graph $H_{\Eq}$ has $L = \sum_{k=1}^{m} \ell_k$ edges.
We order these edges first according to the order $e_1, \ldots, e_m$ of the constraints and then in increasing order of their offsets within each constraint.
For the resulting sequence of edges $\{a_1, b_1\}, \ldots, \{a_L, b_L\}$, we define the permutation associated with $\Eq$ by $\pi_{\Eq} = \tau_{a_1, b_1} \circ \cdots \circ \tau_{a_L, b_L}$.
For each constraint $e = (\{i, j\}, \ell) \in \Eq$, we define the permutation associated with $e$ by $\pi_e = \tau_{i, j} \circ \tau_{i+1, j+1} \circ \cdots \circ \tau_{i+\ell-1, j+\ell-1}$.
Since the edges induced by each constraint are consecutive in this order, $\pi_{\Eq} = \pi_{e_1} \circ \cdots \circ \pi_{e_m}$ holds.
The following lemma shows that the orbits of $\pi_{\Eq}$ correspond to the connected components of $H_{\Eq}$.
\begin{lemma}\label{lem:pi_eq_components}
The orbits of $\pi_{\Eq}$ form the same partition of $[n]$ as the connected components of the referencing graph $H_{\Eq}$.
That is, $\mathcal{O}(\pi_{\Eq}) = \mathcal{C}(H_{\Eq})$.
\end{lemma}
\begin{proof}
For each $0 \leq r \leq L$, let $G_r$ be the graph on $[n]$ with $r$ edges $\{a_1, b_1\}, \ldots, \{a_r, b_r\}$, and let $\pi^{(r)} = \tau_{a_1, b_1} \circ \cdots \circ \tau_{a_r, b_r}$.
In particular, $G_0$ has no edges and $\pi^{(0)}$ is the identity permutation.
We prove $\mathcal{O}(\pi^{(r)}) = \mathcal{C}(G_r)$ by induction on $r$.
The claim holds for $r = 0$ because both partitions consist of singleton sets.

Suppose that the claim holds for $r-1$.
Since $H_{\Eq}$ forms a forest, $a_r$ and $b_r$ belong to different connected components of $G_{r-1}$.
By the induction hypothesis, they also belong to different orbits of $\pi^{(r-1)}$.
After composing $\pi^{(r-1)}$ with $\tau_{a_r, b_r}$, these two orbits form a single orbit of $\pi^{(r)}$.
In $G_r$, the two connected components of $G_{r-1}$ containing $a_r$ and $b_r$ are also merged by the new edge $\{a_r, b_r\}$.
All other orbits and connected components remain unchanged, so $\mathcal{O}(\pi^{(r)}) = \mathcal{C}(G_r)$.
\end{proof}

To bound the number of phrases in the BMS, we count consecutive pairs of positions whose images are not consecutive in the same order.
We denote this set by $\Delta(\pi) = \{p \in [n-1] \mid \pi(p+1) \neq \pi(p) + 1\}$ and its size by $\delta(\pi) = |\Delta(\pi)|$.
We first bound this quantity for a single substring-equality constraint.
\Cref{fig:equality_representing_bijection} illustrates $\pi_e$ for two substring-equality constraints.
\begin{lemma}\label{lem:pi_e_delta}
For every $e \in \Eq$, $\delta(\pi_e) \leq 4$.
\end{lemma}
\begin{proof}
For $e = (\{i, j\}, \ell)$, let $d = j-i$ and write $\ell = qd+r$ with $0 \leq r < d$.
Each transposition in $\pi_e$ exchanges two positions at distance $d$, so $\pi_e(p) \equiv p \pmod d$ holds for every $p \in [n]$.
For each $h \in \{0, \ldots, d-1\}$, we consider the sequence of positions $p_s = i+h+sd$ for $0 \leq s \leq t$, where $t = q+1$ if $h < r$ and $t = q$ otherwise.
The sequence $p_0, \ldots, p_t$ consists of all positions $p \in [i, j+\ell)$ satisfying $p \equiv i+h \pmod d$.
For a fixed $h$, only the transpositions $\tau_{p_0,p_1}, \ldots, \tau_{p_{t-1},p_t}$ in $\pi_e$ change any of $p_0, \ldots, p_t$.
By the order of composition, $\pi_e(p_s) = p_{s+1} = p_s+d$ for $0 \leq s < t$ and $\pi_e(p_t) = p_0 = p_t-td$ hold.

The positions $p_s$ with $s < t$ over all $h$ correspond exactly to $[i, i+\ell)$.
The last positions $p_t = i+h+qd$ for $r \leq h < d$ correspond exactly to $[i+\ell, j+\ell-r)$.
For $0 \leq h < r$, the last positions $p_t = i+h+(q+1)d$ correspond exactly to $[j+\ell-r, j+\ell)$.
Since no transposition in $\pi_e$ changes a position outside $[i, j+\ell)$, we obtain
\[
\pi_e(p) =
\begin{cases}
p+d & \text{if } i \leq p < i+\ell,\\
p-qd & \text{if } i+\ell \leq p < j+\ell-r,\\
p-(q+1)d & \text{if } j+\ell-r \leq p < j+\ell,\\
p & \text{otherwise}.
\end{cases}
\]
The formula gives $\pi_e(p+1) = \pi_e(p)+1$ for every $p \in [n-1] \setminus \{ i-1, i+\ell-1, j+\ell-r-1, j+\ell-1\}$.
Thus, $\delta(\pi_e) \leq 4$.
\end{proof}

\begin{figure}[t]
  \centering
  \includegraphics[width=0.89\textwidth]{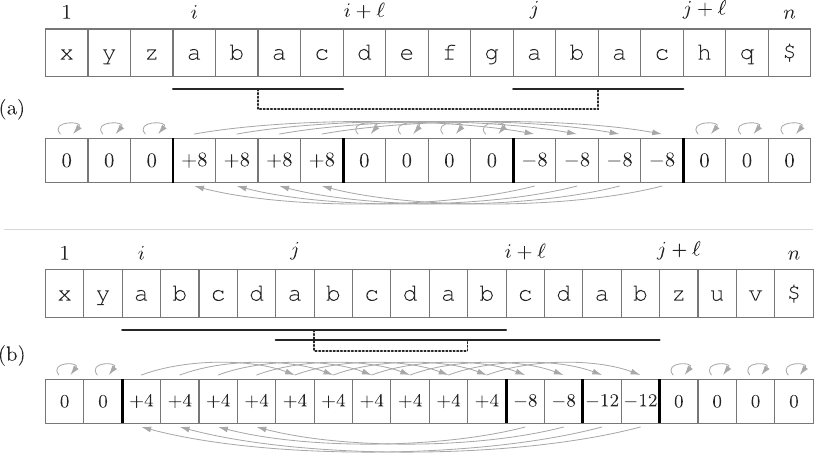}
  \caption{
    Illustrations of $\pi_e$ for substring-equality constraints $e$.
    In each panel, the upper row shows a string and a substring-equality constraint, while each box at position $p$ in the lower row contains $\pi_e(p) - p$.
    The arrows in each lower row show $\pi_e$, with the arrow from each position $p$ pointing to position $\pi_e(p)$.
    Bold lines indicate positions where adjacent values differ, corresponding to $\Delta(\pi_e)$.
    The examples in panels~(a) and~(b) both satisfy $\delta(\pi_e) = 4$.
  }
  \label{fig:equality_representing_bijection}
\end{figure}

The following lemma bounds $\delta$ for a composition of permutations.
\begin{lemma}\label{lem:delta_composition}
For any two permutations $\pi$ and $\rho$ of $[n]$, $\delta(\pi \circ \rho) \leq \delta(\pi) + \delta(\rho)$.
\end{lemma}
\begin{proof}
For any $p \in [n-1] \setminus \Delta(\rho)$ with $\rho(p) \notin \Delta(\pi)$, we have $\rho(p+1) = \rho(p) + 1$ and $\pi(\rho(p) + 1) = \pi(\rho(p)) + 1$.
Hence, $p \notin \Delta(\pi \circ \rho)$, and therefore
$\Delta(\pi \circ \rho) \subseteq \Delta(\rho) \cup \{p \in [n-1] \mid \rho(p) \in \Delta(\pi)\}$.
Since $\rho$ is injective, the second set has at most $|\Delta(\pi)|$ elements.
Thus, $\delta(\pi \circ \rho) \leq \delta(\pi) + \delta(\rho)$.
\end{proof}
By \Cref{lem:pi_e_delta} and \Cref{lem:delta_composition}, $\delta(\pi_{\Eq}) \leq \sum_{k=1}^{m} \delta(\pi_{e_k}) \leq 4m$.

We first consider constructing a BMS whose transition function $\varphi$ satisfies $\varphi(p) = \pi_{\Eq}(p)$ for every $p \in [n]$
and each phrase corresponds to a maximal interval $[l, r]$ satisfying $\pi_{\Eq}(p+1) = \pi_{\Eq}(p)+1$ for every $l \leq p < r$.
Since $\delta(\pi_{\Eq})$ counts the positions $p \in [n-1]$ satisfying $\pi_{\Eq}(p+1) \neq \pi_{\Eq}(p)+1$, the number of phrases is at most $\delta(\pi_{\Eq})+1 \leq 4m+1$.
The reference relation forms a cycle on each orbit of $\pi_{\Eq}$, and each orbit corresponds to a connected component of $H_{\Eq}$.
To make this BMS valid, we select a representative $p_C \in C$ for each connected component $C \in \mathcal{C}(H_{\Eq})$ and make $w[p_C..p_C]$ a phrase that stores the character $w[p_C]$ explicitly.
We denote the set of representatives by $S = \{p_C \mid C \in \mathcal{C}(H_{\Eq})\}$.
Let $\mathcal{I}$ be the set of maximal nonempty intervals $[l, r] \subseteq [n] \setminus S$ such that $\pi_{\Eq}(p+1) = \pi_{\Eq}(p) + 1$ for every $l \leq p < r$.
By definition, $\mathcal{I} \cup \{\{p\} \mid p \in S\}$ forms a partition of $[n]$.
For each $[l, r] \in \mathcal{I}$, we have $\pi_{\Eq}(p) = \pi_{\Eq}(l) + p-l$ for every $p \in [l, r]$.

Formally, we define the BMS $\mathcal{B}$ as follows:
\begin{itemize}
  \item For each $p \in S$, let $w[p..p]$ be a phrase representing the single character $w[p]$.
  \item For each $[l, r] \in \mathcal{I}$, let $w[l..r]$ be a phrase whose source interval is $[\pi_{\Eq}(l), \pi_{\Eq}(r)]$.
\end{itemize}
\Cref{fig:tribonacci_ses_to_bms} illustrates an example of the BMS constructed from an SES.

\begin{figure}[t]
  \centering
  \includegraphics[width=0.98\textwidth]{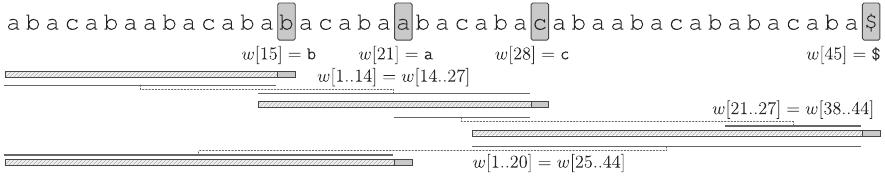}
  \par\vspace{0.6\baselineskip}
  \includegraphics[width=0.98\textwidth]{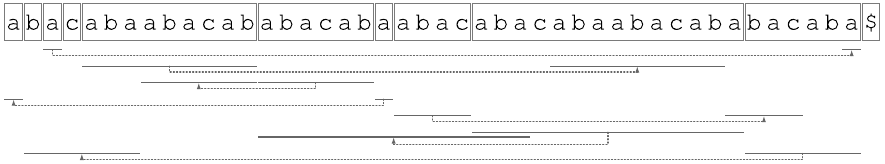}
  \caption{
    An example of an SES derived from the super-maximal right extensions of the Tribonacci string $T_6$ with the unique character $\$$ and a BMS constructed from this SES.
    The Tribonacci string is defined as $T_0 = \mathtt{a}$, $T_1 = \mathtt{ab}$, $T_2 = \mathtt{abac}$, and $T_k = T_{k-1} T_{k-2} T_{k-3}$ for $k \geq 3$.
    In the upper diagram, the rounded boxes mark the positions in a smallest suffixient set, with the character-assignment constraints written below them.
    Each block shows the leftmost occurrence of a super-maximal right extension $xc$, with its hatched and solid gray segments marking $x$ and $c$, respectively.
    Each pair of solid segments connected by a dotted line represents the substring-equality constraint written beside it.
    In the lower diagram, each box represents one phrase of the BMS constructed from this SES, with a dotted arrow from each copied phrase to its source interval.
  }
  \label{fig:tribonacci_ses_to_bms}
\end{figure}

\begin{theorem}
The BMS $\mathcal{B}$ is valid for $w$
and its size is at most $4(|\Eq| + |\Ch|)$.
\end{theorem}
\begin{proof}
Let $\varphi_{\mathcal{B}}: [n] \to [n] \cup \{\bot\}$ be the transition function induced by $\mathcal{B}$.
By construction, $\varphi_{\mathcal{B}}(p) = \bot$ if $p \in S$ and $\varphi_{\mathcal{B}}(p) = \pi_{\Eq}(p)$ otherwise.
To prove that $\mathcal{B}$ is a valid BMS for $w$, it suffices to show that $w[p] = w[\varphi_{\mathcal{B}}(p)]$ for every $p \in [n]$ such that $\varphi_{\mathcal{B}}(p) \neq \bot$ and that for every position $p \in [n]$ there exists an integer $k_p \geq 0$ such that $\varphi_{\mathcal{B}}^{k_p}(p) = \bot$.
These properties are immediate when $\varphi_{\mathcal{B}}(p) = \bot$.

Consider a position $p \in [n]$ such that $\varphi_{\mathcal{B}}(p) \neq \bot$.
Let $C$ be the connected component of $H_{\Eq}$ that contains $p$.
By \Cref{lem:pi_eq_components}, $C = O_{\pi_{\Eq}}(p) = \{\pi_{\Eq}^t(p) \mid t \geq 0\}$.
Because $p$ and $\pi_{\Eq}(p)$ belong to $C$, $w[p] = w[\pi_{\Eq}(p)] = w[\varphi_{\mathcal{B}}(p)]$.
This proves the first property.

By definition, the connected component $C$ contains the unique representative $p_C \in C \cap S$.
This representative $p_C$ belongs to the orbit of $p$ because $C = O_{\pi_{\Eq}}(p)$.
Since $p_C$ belongs to the orbit of $p$ and $\varphi_{\mathcal B}(x) = \pi_{\Eq}(x)$ for every $x \in C \setminus \{p_C\}$, 
there exists an integer $k$ such that $\varphi_{\mathcal{B}}^{k}(p) = p_C$
and thus $\varphi_{\mathcal{B}}^{k+1}(p) = \varphi_{\mathcal{B}}(p_C) = \bot$, which proves the second property.

We next bound the number of phrases in $\mathcal{B}$.
By construction, $\mathcal{B}$ has exactly $|\mathcal{I}| + |S|$ phrases.
There are $\delta(\pi_{\Eq}) + 1$ maximal nonempty intervals $[l, r] \subseteq [n]$ such that $\pi_{\Eq}(p+1) = \pi_{\Eq}(p) + 1$ for every $l \leq p < r$.
Removing each position in $S$ increases the number of intervals by at most one.
Therefore, $|\mathcal{I}| \leq \delta(\pi_{\Eq}) + 1 + |S| \leq 4m + |S| + 1$.
Moreover, $|S| = |\mathcal{C}(H_\Eq)| = \kappa(\Eq)$.
Each connected component of $H_{\Eq}$ contains at least one position with a character-assignment constraint.
Hence, $|S| = \kappa(\Eq) \leq |\Ch|$.
By combining these inequalities, we obtain $|\mathcal{I}| + |S| \leq 4m + 2|S| + 1 \leq 4 |\Eq| + 2 |\Ch| + 1$.
Since $w$ is nonempty, $|\Ch| \geq 1$, and hence $4 |\Eq| + 2 |\Ch| + 1 \leq 4(|\Eq| + |\Ch|)$.
\end{proof}
 
\section*{AI Usage Disclosure}
The authors used OpenAI's ChatGPT (GPT-5.5, GPT-5.6-Sol, and GPT-6-Astra) for drafting technical text, correcting typographical errors, improving the organization of the manuscript, and exploring proof strategies.
In particular, the models assisted in designing the gadgets for the reduction and proving the hardness results in \Cref{se:np_hardness}, as well as in developing the proof strategy and proving the main theorem in \Cref{se:bms_ses_equivalence}.
The authors verified all claims and proofs and take full responsibility for the contents of this paper.
 \bibliography{refs}

@inproceedings{DBLP:conf/spire/NavarroRU25,
  author       = {Gonzalo Navarro and
                  Giuseppe Romana and
                  Cristian Urbina},
  editor       = {Golnaz Badkobeh and
                  Jakub Radoszewski and
                  Nicola Tonellotto and
                  Ricardo Baeza{-}Yates},
  title        = {Smallest Suffixient Sets as a Repetitiveness Measure},
  booktitle    = {String Processing and Information Retrieval - 32nd International Symposium,
                  {SPIRE} 2025, London, UK, September 8-11, 2025, Proceedings},
  series       = {Lecture Notes in Computer Science},
  volume       = {16073},
  pages        = {217--232},
  publisher    = {Springer},
  year         = {2025},
  doi          = {10.1007/978-3-032-05228-5_18},
  bibsource    = {dblp computer science bibliography, https://dblp.org}
}

@article{DBLP:journals/tcs/GawrychowskiKRR20,
  author       = {Pawe{\l} Gawrychowski and
                  Tomasz Kociumaka and
                  Jakub Radoszewski and
                  Wojciech Rytter and
                  Tomasz Wale{\'{n}}},
  title        = {Universal reconstruction of a string},
  journal      = {Theor. Comput. Sci.},
  volume       = {812},
  pages        = {174--186},
  year         = {2020},
  doi          = {10.1016/J.TCS.2019.10.027},
  bibsource    = {dblp computer science bibliography, https://dblp.org}
}

@article{DBLP:journals/csur/Navarro21a,
  author       = {Gonzalo Navarro},
  title        = {Indexing Highly Repetitive String Collections, Part {I:} Repetitiveness
                  Measures},
  journal      = {{ACM} Comput. Surv.},
  volume       = {54},
  number       = {2},
  pages        = {29:1--29:31},
  year         = {2022},
  doi          = {10.1145/3434399},
  bibsource    = {dblp computer science bibliography, https://dblp.org}
}

@article{DBLP:journals/corr/abs-2312-01359,
  author       = {Lore Depuydt and
                  Travis Gagie and
                  Ben Langmead and
                  Giovanni Manzini and
                  Nicola Prezza},
  title        = {Suffixient Sets},
  journal      = {CoRR},
  volume       = {abs/2312.01359},
  year         = {2023},
  doi          = {10.48550/ARXIV.2312.01359},
  bibsource    = {dblp computer science bibliography, https://dblp.org}
}

@inproceedings{DBLP:conf/spire/CenzatoOP24,
  author       = {Davide Cenzato and
                  Francisco Olivares and
                  Nicola Prezza},
  editor       = {Zsuzsanna Lipt{\'{a}}k and
                  Edleno Silva de Moura and
                  Karina Figueroa and
                  Ricardo Baeza{-}Yates},
  title        = {On Computing the Smallest Suffixient Set},
  booktitle    = {String Processing and Information Retrieval - 31st International Symposium,
                  {SPIRE} 2024, Puerto Vallarta, Mexico, September 23-25, 2024, Proceedings},
  series       = {Lecture Notes in Computer Science},
  volume       = {14899},
  pages        = {73--87},
  publisher    = {Springer},
  year         = {2024},
  doi          = {10.1007/978-3-031-72200-4_6},
  bibsource    = {dblp computer science bibliography, https://dblp.org}
}

@techreport{Burrows1994ABL,
  author       = {Michael Burrows and
                  David J. Wheeler
                  },
  title        = {A Block-sorting Lossless Data Compression Algorithm},
  type         = {SRC Research Report},
  number       = {124},
  institution  = {Digital Equipment Corporation, Systems Research Center},
  year         = {1994},
}

@article{DBLP:journals/tit/ZivL77,
  author       = {Jacob Ziv and
                  Abraham Lempel},
  title        = {A universal algorithm for sequential data compression},
  journal      = {{IEEE} Trans. Inf. Theory},
  volume       = {23},
  number       = {3},
  pages        = {337--343},
  year         = {1977},
  doi          = {10.1109/TIT.1977.1055714},
  bibsource    = {dblp computer science bibliography, https://dblp.org}
}

@inproceedings{DBLP:conf/stoc/StorerS78,
  author       = {James A. Storer and
                  Thomas G. Szymanski},
  editor       = {Richard J. Lipton and
                  Walter A. Burkhard and
                  Walter J. Savitch and
                  Emily P. Friedman and
                  Alfred V. Aho},
  title        = {The Macro Model for Data Compression (Extended Abstract)},
  booktitle    = {Proceedings of the 10th Annual {ACM} Symposium on Theory of Computing,
                  May 1-3, 1978, San Diego, California, {USA}},
  pages        = {30--39},
  publisher    = {{ACM}},
  year         = {1978},
  doi          = {10.1145/800133.804329},
  bibsource    = {dblp computer science bibliography, https://dblp.org}
}

@inproceedings{DBLP:conf/stoc/KempaP18,
  author       = {Dominik Kempa and
                  Nicola Prezza},
  editor       = {Ilias Diakonikolas and
                  David Kempe and
                  Monika Henzinger},
  title        = {At the Roots of Dictionary Compression: String Attractors},
  booktitle    = {Proceedings of the 50th Annual {ACM} {SIGACT} Symposium on Theory
                  of Computing, {STOC} 2018, Los Angeles, CA, USA, June 25-29, 2018},
  pages        = {827--840},
  publisher    = {{ACM}},
  year         = {2018},
  doi          = {10.1145/3188745.3188814},
  bibsource    = {dblp computer science bibliography, https://dblp.org}
}

@inproceedings{DBLP:conf/latin/KociumakaNP20,
  author       = {Tomasz Kociumaka and
                  Gonzalo Navarro and
                  Nicola Prezza},
  editor       = {Yoshiharu Kohayakawa and
                  Fl{\'{a}}vio Keidi Miyazawa},
  title        = {Towards a Definitive Measure of Repetitiveness},
  booktitle    = {{LATIN} 2020: Theoretical Informatics - 14th Latin American Symposium,
                  S{\~{a}}o Paulo, Brazil, January 5-8, 2021, Proceedings},
  series       = {Lecture Notes in Computer Science},
  volume       = {12118},
  pages        = {207--219},
  publisher    = {Springer},
  year         = {2020},
  doi          = {10.1007/978-3-030-61792-9_17},
  bibsource    = {dblp computer science bibliography, https://dblp.org}
}

@article{DBLP:journals/mst/CenzatoDGKMOP26,
  author       = {Davide Cenzato and
                  Lore Depuydt and
                  Travis Gagie and
                  Sung{-}Hwan Kim and
                  Giovanni Manzini and
                  Francisco Olivares and
                  Nicola Prezza},
  title        = {Suffixient Arrays: {A} New Efficient Suffix Array Compression Technique},
  journal      = {Theory Comput. Syst.},
  volume       = {70},
  number       = {3},
  pages        = {50},
  year         = {2026},
  doi          = {10.1007/S00224-026-10287-6},
  bibsource    = {dblp computer science bibliography, https://dblp.org}
}

@inproceedings{DBLP:conf/cpm/KasaiLAAP01,
  author       = {Toru Kasai and
                  Gunho Lee and
                  Hiroki Arimura and
                  Setsuo Arikawa and
                  Kunsoo Park},
  editor       = {Amihood Amir and
                  Gad M. Landau},
  title        = {Linear-Time Longest-Common-Prefix Computation in Suffix Arrays and
                  Its Applications},
  booktitle    = {Combinatorial Pattern Matching, 12th Annual Symposium, {CPM} 2001
                  Jerusalem, Israel, July 1-4, 2001 Proceedings},
  series       = {Lecture Notes in Computer Science},
  volume       = {2089},
  pages        = {181--192},
  publisher    = {Springer},
  year         = {2001},
  doi          = {10.1007/3-540-48194-X_17},
  bibsource    = {dblp computer science bibliography, https://dblp.org}
}

@inproceedings{DBLP:conf/focs/Farach97,
  author       = {Martin Farach},
  title        = {Optimal Suffix Tree Construction with Large Alphabets},
  booktitle    = {38th Annual Symposium on Foundations of Computer Science, {FOCS} 1997,
                  Miami Beach, Florida, USA, October 19-22, 1997},
  pages        = {137--143},
  publisher    = {{IEEE} Computer Society},
  year         = {1997},
  doi          = {10.1109/SFCS.1997.646102},
  bibsource    = {dblp computer science bibliography, https://dblp.org}
}

@article{DBLP:journals/tcs/ChlebikC06,
  author       = {Miroslav Chleb{\'{\i}}k and
                  Janka Chleb{\'{\i}}kov{\'{a}}},
  title        = {Complexity of approximating bounded variants of optimization problems},
  journal      = {Theor. Comput. Sci.},
  volume       = {354},
  number       = {3},
  pages        = {320--338},
  year         = {2006},
  doi          = {10.1016/J.TCS.2005.11.029},
  bibsource    = {dblp computer science bibliography, https://dblp.org}
}

@article{DBLP:journals/siamcomp/ManberM93,
  author       = {Udi Manber and
                  Eugene W. Myers},
  title        = {Suffix Arrays: {A} New Method for On-Line String Searches},
  journal      = {{SIAM} J. Comput.},
  volume       = {22},
  number       = {5},
  pages        = {935--948},
  year         = {1993},
  doi          = {10.1137/0222058},
  bibsource    = {dblp computer science bibliography, https://dblp.org}
}

@inproceedings{DBLP:conf/soda/KociumakaNW25,
  author       = {Tomasz Kociumaka and
                  Jakob Nogler and
                  Philip Wellnitz},
  editor       = {Yossi Azar and
                  Debmalya Panigrahi},
  title        = {Near-Optimal-Time Quantum Algorithms for Approximate Pattern Matching},
  booktitle    = {Proceedings of the 2025 Annual {ACM-SIAM} Symposium on Discrete Algorithms,
                  {SODA} 2025, New Orleans, LA, USA, January 12-15, 2025},
  pages        = {517--534},
  publisher    = {{SIAM}},
  year         = {2025},
  doi          = {10.1137/1.9781611978322.15},
  bibsource    = {dblp computer science bibliography, https://dblp.org}
}

@article{DBLP:journals/jacm/StorerS82,
  author       = {James A. Storer and
                  Thomas G. Szymanski},
  title        = {Data compression via textual substitution},
  journal      = {J. {ACM}},
  volume       = {29},
  number       = {4},
  pages        = {928--951},
  year         = {1982},
  doi          = {10.1145/322344.322346},
  bibsource    = {dblp computer science bibliography, https://dblp.org}
}

@inproceedings{DBLP:conf/esa/KempaPPR18,
  author       = {Dominik Kempa and
                  Alberto Policriti and
                  Nicola Prezza and
                  Eva Rotenberg},
  editor       = {Yossi Azar and
                  Hannah Bast and
                  Grzegorz Herman},
  title        = {String Attractors: Verification and Optimization},
  booktitle    = {26th Annual European Symposium on Algorithms, {ESA} 2018, Helsinki,
                  Finland, August 20-22, 2018},
  series       = {LIPIcs},
  volume       = {112},
  pages        = {52:1--52:13},
  publisher    = {Schloss Dagstuhl - Leibniz-Zentrum f{\"{u}}r Informatik},
  year         = {2018},
  doi          = {10.4230/LIPICS.ESA.2018.52},
  bibsource    = {dblp computer science bibliography, https://dblp.org}
}

@inproceedings{DBLP:conf/icalp/PlandowskiR98,
  author       = {Wojciech Plandowski and
                  Wojciech Rytter},
  editor       = {Kim Guldstrand Larsen and
                  Sven Skyum and
                  Glynn Winskel},
  title        = {Application of {Lempel-Ziv} Encodings to the Solution of Words Equations},
  booktitle    = {Automata, Languages and Programming, 25th International Colloquium,
                  ICALP'98, Aalborg, Denmark, July 13-17, 1998, Proceedings},
  series       = {Lecture Notes in Computer Science},
  volume       = {1443},
  pages        = {731--742},
  publisher    = {Springer},
  year         = {1998},
  doi          = {10.1007/BFB0055097},
  bibsource    = {dblp computer science bibliography, https://dblp.org}
}

@article{AgolHT06,
  author       = {Ian Agol and
                  Joel Hass and
                  William Thurston},
  title        = {The computational complexity of knot genus and spanning area},
  journal      = {Trans. Amer. Math. Soc.},
  volume       = {358},
  number       = {9},
  pages        = {3821--3850},
  year         = {2006},
  doi          = {10.1090/S0002-9947-05-03919-X}
}

\end{document}